\documentclass[pdflatex,sn-mathphys-ay]{sn-jnl}

\usepackage{graphicx}%
\usepackage{multirow}%
\usepackage{amsmath,amssymb,amsfonts}%
\usepackage{amsthm}%
\usepackage{mathrsfs}%
\usepackage[title]{appendix}%
\usepackage{xcolor}%
\usepackage{textcomp}%
\usepackage{manyfoot}%
\usepackage{booktabs}%
\usepackage{makecell}%
\usepackage{algorithm}%
\usepackage{listings}%
\usepackage{algorithmic}%
\usepackage{xspace}%
\usepackage{enumitem}

\usepackage{caption}
\usepackage{subcaption}

\theoremstyle{thmstyleone}%
\newtheorem{theorem}{Theorem}
\newtheorem{proposition}[theorem]{Proposition}%
\newtheorem{lemma}[theorem]{Lemma}
\newtheorem{corollary}[theorem]{Corollary}

\theoremstyle{thmstyletwo}%

\theoremstyle{thmstylethree}%

\DeclareMathOperator*{\argmin}{arg\,min}
\newcommand{\R}{\mathbb R}
\newcommand{\calP}{\mathcal P}
\newcommand{\calM}{\mathcal M}
\newcommand{\calB}{\mathcal B}
\newcommand{\calA}{\mathcal A}
\newcommand{\eps}{\varepsilon}
\newcommand{\barmu}{\bar\mu}
\newcommand{\Med}{\operatorname{Med}}
\newcommand{\dist}{\operatorname{dist}}
\newcommand{\conv}{\operatorname{conv}}
\newcommand{\norm}[1]{\left\lVert #1\right\rVert}

\newcommand{\PLoja}{Polyak--\L{}ojasiewicz\xspace}
\newcommand{\KLoja}{Kurdyka--\L{}ojasiewicz\xspace}
\begin{document}

\title[Fast Computation of Free-Support Wasserstein Medians]{Fast Computation of Free-Support Wasserstein Medians}



\author[1,2]{\fnm{Kisung} \sur{You}}
\author[3,4]{\fnm{Mauro} \sur{Giuffr\`{e}} }
\author[5]{\fnm{Dennis} \sur{Shung}}


\affil*[1]{\orgdiv{Department of Mathematics}, \orgname{Baruch College}, \orgaddress{\street{55 Lexington Street}, \city{New York}, \postcode{10010}, \state{NY}, \country{USA}}}

\affil*[2]{\orgdiv{Department of Mathematics}, \orgname{The Graduate Center, City University of New York}, \orgaddress{\street{365 5th Ave}, \city{New York}, \postcode{10016}, \state{NY}, \country{USA}}}

\affil[3]{\orgdiv{Department of Biomedical Informatics and Data Science}, \orgname{Yale School of Medicine}, \orgaddress{\street{101 College Street}, \city{New Haven}, \postcode{06510}, \state{CT}, \country{USA}}}

\affil[4]{\orgdiv{Department of Medicine, Surgery, and Health Sciences}, \orgname{University of Trieste}, \orgaddress{\street{Strada di Fiume 447}, \city{Trieste}, \postcode{34149}, \state{TS}, \country{Italy}}}

\affil[5]{\orgdiv{Division of Gastroenterology and Hepatology, Department of Medicine}, \orgname{Mayo Clinic}, \orgaddress{\street{200 First St SW}, \city{Rochester}, \postcode{55905}, \state{MN}, \country{USA}}}


\abstract{
The Wasserstein median is a robust alternative to the Wasserstein barycenter for averaging probability measures, but exact empirical computation can be expensive. A natural metric-space Weiszfeld scheme updates the current candidate by solving a weighted Wasserstein barycenter problem at each outer iteration, producing a nested optimization problem. We propose a direct fixed-weight free-support solver that avoids this inner barycenter loop. At each iteration, the method solves exact optimal transport (OT) subproblems from the current candidate to the input measures, computes barycentric projections of the selected plans, and relocates each support atom to an inverse-distance-weighted average of its projected destinations. For a smoothed median objective, we show that this relocation is the exact minimizer of a tight majorization--minimization surrogate. This yields monotone descent for exact transport subproblems, convex-hull invariance, a finite-time best-residual rate, residual-to-gradient control under differentiability, and fixed-point and stationarity characterizations. We also give smoothing, stability, and resolution-consistency results clarifying the fixed-weight approximation. In exact-OT benchmarks, the direct solver attains median objectives close to tightly solved nested Weiszfeld baselines while using substantially fewer exact transport subproblems. Additional contamination, posterior aggregation, and image-prototype experiments show that the direct solver produces median summaries comparable to nested computation and less sensitive to outlying distributions than Wasserstein barycenters.
}

\keywords{Optimal transport, Wasserstein median, Weiszfeld algorithm, Majorization--minimization, Robust aggregation}



\maketitle

\section{Introduction}

The Wasserstein barycenter is a widely used tool for averaging probability measures in the literature of optimal transport (OT). It is the Fr\'{e}chet mean under the Wasserstein distance and minimizes a sum of squared Wasserstein distances \citep{agueh_2011_BarycentersWassersteinSpace}. Like the Euclidean mean, it is geometrically natural but mean-like. For instance, when some input distributions are outlying, a barycenter can be pulled toward them.

The Wasserstein median, as a robust alternative to the barycenter, replaces the squared-distance objective by a distance objective. For probability measures \(\mu_1,\ldots,\mu_N\) and positive weights \(\pi_n\) summing to one, the median objective is
\begin{equation}\label{eq:intro-median}
    \mu \mapsto \sum_{n=1}^N \pi_n W_2(\mu,\mu_n).
\end{equation}
\citet{you_2025_WassersteinMedianProbability} established existence, consistency, robustness properties, and a general computational meta-algorithm for the Wasserstein median based on iteratively solving Wasserstein barycenter problems. This approach can be cast as a standard Weiszfeld algorithm for computing the geometric median \citep{weiszfeld_1937_PointPourLequel}.

The computational bottleneck is easy to see from the metric-space Weiszfeld viewpoint. Starting from a candidate \(\mu^{(k)}\), a Weiszfeld iteration forms inverse-distance weights
\[
    \lambda_n^{(k)} \propto \frac{\pi_n}{W_2(\mu^{(k)},\mu_n)}
\]
and then updates \(\mu^{(k+1)}\) by solving a weighted Wasserstein barycenter problem,
\[
    \mu^{(k+1)} \in \argmin_{\mu}\sum_{n=1}^N \lambda_n^{(k)} W_2^2(\mu,\mu_n).
\]
For empirical measures whose support is not constrained, this barycenter problem is itself an iterative optimization problem. Thus an outer median iteration contains an inner barycenter computation. If the inner barycenter routine uses \(L_k\) iterations at outer step \(k\), the nested method requires roughly \(N L_k\) exact OT subproblems during that outer step.

We propose a direct fixed-weight free-support solver that avoids this inner loop. At the current candidate support, the method solves the same \(N\) exact OT subproblems needed to evaluate the Weiszfeld distances, computes barycentric projections of the selected plans, and relocates each atom to an inverse-distance-weighted average of its projected destinations. Operationally, the step resembles one warm-started inner barycenter iteration. Its justification is different: freezing the current OT plans gives a feasibility-invariant quadratic upper bound on each squared transport cost, and linearizing the smoothed square-root distance gives Weiszfeld weights. The resulting combined surrogate is a majorization--minimization (MM) surrogate for the median objective itself, and its closed-form minimizer is precisely the proposed relocation.

The construction combines established ideas in a new free-support median setting. Weiszfeld's method and its singularity issues go back to \citet{weiszfeld_1937_PointPourLequel}; \citet{kuhn_1973_NoteFermatsProblem} and \citet{beck_2015_WeiszfeldsMethodOld} analyze classical anchor-point and smoothing issues. Barycentric-projection fixed-point ideas for Wasserstein barycenters appear in \citet{alvarez-esteban_2016_FixedpointApproachBarycenters} and \citet{cuturi_2014_FastComputationWasserstein}. Our contribution is to combine plan-freezing and smoothed Weiszfeld reweighting to obtain a direct exact-OT solver for fixed-weight free-support Wasserstein medians.

The method is intentionally restricted to a fixed-weight free-support class. A candidate median has the form
\[
    \barmu_Z=\sum_{i=1}^m v_i\delta_{z_i},\qquad v_i>0,
\]
where the weights \(v_i\) are chosen in advance, usually uniformly, and only the locations \(z_i\) are optimized. This is not the full free-support median over both locations and weights at a fixed \(m\). Adaptive support weights are an important extension, but they are outside the scope of this paper. The resolution result later in the paper explains how increasingly rich fixed-weight approximation classes can recover the population median set in the limit \(m\to\infty\).

For \(\eps>0\), we prove that the full relocation step is the exact minimizer of a tight frozen-plan quadratic surrogate. This gives monotone descent for any selected exact optimal plans, without uniqueness of the transport plans. The iterates enter the convex hull of the input supports after one step, which yields residual convergence and an \(O(1/K)\) best-residual rate. Under differentiability of the discrete OT value functions, the residual controls the gradient, and accumulation points become stationary under a standard Danskin-type condition. We also give smoothing, stability, and resolution-consistency results to relate the smoothed fixed-weight discrete problem to the nonsmoothed Wasserstein median objective.

The computational claims of the paper are deliberately exact OT claims. We do not compare with entropic or Sinkhorn methods in the main design. The central comparison is against exact nested Wasserstein--Weiszfeld baselines: a tight nested solve, and nested variants with fixed inner barycenter budgets. The experiments are designed to report wall-clock time, exact OT subproblem counts, objective gaps, residual histories, and scaling with the number of measures, support size, and dimension.

The rest of the paper is organized as follows. Section~\ref{sec:prelim} reviews the Wasserstein distance, empirical OT as a linear program, barycentric projections, and the distinction between Wasserstein barycenters and medians. Section~\ref{sec:formulation} formulates the fixed-weight free-support median problem and the exact nested Wasserstein--Weiszfeld baseline. Section~\ref{sec:algorithm} presents the direct Wasserstein--Weiszfeld particle solver. Section~\ref{sec:theory} gives the MM analysis, residual guarantees, stationarity statements, smoothing bounds, objective stability, and resolution consistency. Section~\ref{sec:implementation} discusses exact OT cost accounting and implementation diagnostics. Section~\ref{sec:experiments} reports the computational experiments, and Section~\ref{sec:discussion} concludes with limitations and future directions.

\section{Preliminaries}\label{sec:prelim}

This section reviews the notation and finite-dimensional OT facts used by the solver. Throughout the paper, \(\norm{\cdot}\) denotes the Euclidean norm on \(\R^d\).

\subsection{The 2-Wasserstein distance}

Let \(\calP_2(\R^d)\) denote the set of Borel probability measures on \(\R^d\) with finite second moment. For \(\mu,\nu\in\calP_2(\R^d)\), the squared 2-Wasserstein distance is
\begin{equation}\label{eq:w2}
    W_2^2(\mu,\nu)
    =
    \inf_{\gamma\in\Pi(\mu,\nu)}
    \int_{\R^d\times\R^d}\norm{x-y}^2\,d\gamma(x,y),
\end{equation}
where \(\Pi(\mu,\nu)\) is the set of couplings with marginals \(\mu\) and \(\nu\). The metric \(W_2\) turns \(\calP_2(\R^d)\) into a natural geometry for comparing probability measures. For the rest of this paper, the 2-Wasserstein distance will be denoted simply the Wasserstein distance. See  \citet{villani_2009_OptimalTransportOld} and \citet{ambrosio_2021_LecturesOptimalTransport} for background.

\subsection{Empirical optimal transport as a linear program}

The algorithms in this paper work with empirical measures. Suppose
\[
    \mu=\sum_{i=1}^m a_i\delta_{z_i},
    \qquad
    \nu=\sum_{j=1}^{m'} b_j\delta_{x_j},
\]
where \(a_i,b_j>0\) and the weights sum to one. Then \(W_2^2(\mu,
\nu)\) is the finite linear program
\begin{equation}\label{eq:empirical-lp}
    W_2^2(\mu,\nu)
    =
    \min_{\Gamma\in\Pi(a,b)}
    \sum_{i=1}^{m}\sum_{j=1}^{m'}
    \Gamma_{ij}\norm{z_i-x_j}^2,
\end{equation}
where
\[
    \Pi(a,b)=\{\Gamma\in\R_+^{m\times m'}:\Gamma\mathbf 1_{m'}=a,
    \Gamma^\top\mathbf 1_m=b\}.
\]
The transport polytope \(\Pi(a,b)\) depends only on the weights, not on the support locations. Thus if a plan is optimal for one support configuration, it remains feasible after the source support moves. This simple feasibility-invariance observation is the key to the MM construction below.

\subsection{Barycentric projections}

Given a plan \(\Gamma\in\Pi(a,b)\), the barycentric projection from source atom \(z_i\) is
\[
    B(z_i)=\frac{1}{a_i}\sum_{j=1}^{m'}\Gamma_{ij}x_j.
\]
This is a conditional mean of target locations under the selected plan. While it is not strictly a ``map'' in the sense of \citet{monge_1781_MemoireTheorieDeblais}, barycentric projections nevertheless provide a practical way to turn an optimal plan obtained from \eqref{eq:empirical-lp} into a destination for each free-support atom. They are central to free-support barycenter algorithms \citep{cuturi_2014_FastComputationWasserstein, alvarez-esteban_2016_FixedpointApproachBarycenters} and to the median solver proposed here.

\subsection{Wasserstein barycenters and medians}

The Wasserstein barycenter minimizes
\[
    \calB(\mu)=\sum_{n=1}^N \pi_n W_2^2(\mu,\mu_n),
\]
whereas the Wasserstein median minimizes
\begin{equation}\label{eq:population-median}
    \calM(\mu)=\sum_{n=1}^N \pi_n W_2(\mu,\mu_n).
\end{equation}
The barycenter is mean-like, while the median is median-like. This distinction mirrors the Euclidean difference between least squares and least absolute deviations. The Wasserstein median functional, its robustness, and a nested computational pipeline were studied by \citet{you_2025_WassersteinMedianProbability}. The present paper assumes this target and focuses on accelerating exact empirical computation.

If \(K\subset\R^d\) is compact, we write \(\calP(K)\) for probability measures supported in \(K\). Since \(K\) is compact, all such measures have finite second moment, so \(\calP(K)\subset\calP_2(\R^d)\), and the topology induced by \(W_2\) agrees with weak convergence on \(\calP(K)\) \citep{villani_2003_TopicsOptimalTransportation, panaretos_2020_InvitationStatisticsWasserstein}.


\section{Fixed-weight formulation and exact nested baseline}\label{sec:formulation}

\subsection{The fixed-weight free-support class}

The inputs are empirical measures
\[
    \mu_n=\sum_{j=1}^{m_n}w_{n,j}\delta_{x_{n,j}},
    \qquad n=1,\ldots,N.
\]
We approximate the median by a measure with fixed positive weights
\begin{equation}\label{eq:candidate}
    \barmu_Z=\sum_{i=1}^m v_i\delta_{z_i},
    \qquad v_i>0,
    \qquad \sum_{i=1}^m v_i=1.
\end{equation}
Only the locations \(Z=(z_1,\ldots,z_m)\in(\R^d)^m\) are optimized. The fixed-weight nonsmoothed free-support median objective is
\begin{equation}\label{eq:discrete-median}
    \calM_m(Z)=\sum_{n=1}^N \pi_n W_2(\barmu_Z,\mu_n).
\end{equation}
When no current distance vanishes, the exact Weiszfeld weights for \(\calM_m\) are well defined. To obtain a globally finite and differentiable distance surrogate, we use the smoothed objective
\begin{equation}\label{eq:smoothed-objective}
    \calM_{m,\eps}(Z)
    =
    \sum_{n=1}^N\pi_n
    \sqrt{W_2^2(\barmu_Z,\mu_n)+\eps^2},
    \qquad \eps>0.
\end{equation}
The smoothed objective is the object analyzed in the main theorems. The exact case \(\eps=0\) will be discussed later. 

\subsection{Exact nested Wasserstein--Weiszfeld}

At a current support \(Z^{(k)}\), define
\[
    d_n^{(k)}=W_2(\barmu_{Z^{(k)}},\mu_n),
    \qquad
    r_{n,\eps}^{(k)}=\sqrt{(d_n^{(k)})^2+\eps^2}.
\]
The smoothed Weiszfeld weights are
\begin{equation}\label{eq:lambda}
    \lambda_{n,\eps}^{(k)}
    =
    \frac{\pi_n/r_{n,\eps}^{(k)}}
    {\sum_{\ell=1}^N\pi_\ell/r_{\ell,\eps}^{(k)}}.
\end{equation}
A Weiszfeld iteration then updates by solving the weighted Wasserstein barycenter problem
\begin{equation}\label{eq:nested-barycenter}
    Z^{(k+1)}\approx
    \argmin_Z
    \sum_{n=1}^N \lambda_{n,\eps}^{(k)} W_2^2(\barmu_Z,\mu_n).
\end{equation}
This is the exact nested baseline. It corresponds to the computational pipeline in \citet{you_2025_WassersteinMedianProbability} when the barycenter subroutine is an exact OT free-support solver. If the inner solver uses \(L_k\) barycenter iterations at outer iteration \(k\), then the outer step requires approximately \(N L_k\) exact OT subproblems.

\section{Direct Wasserstein--Weiszfeld particle solver}\label{sec:algorithm}

The proposed solver uses the same distances and weights as the nested method but replaces the inner barycenter solve by a closed-form relocation step. At iteration \(k\), for each \(n\), choose an exact optimal plan
\[
    \Gamma^{(n,k)}\in\Pi(v,w_n)
\]
from \(\barmu_{Z^{(k)}}\) to \(\mu_n\). Define the barycentric projection
\begin{equation}\label{eq:baryproj}
    B_n^{(k)}(z_i^{(k)})
    =
    \frac{1}{v_i}\sum_{j=1}^{m_n}\Gamma_{ij}^{(n,k)}x_{n,j}.
\end{equation}
The denominator is well defined because the support weights \(v_i\) are fixed and strictly positive. The proposed destination for atom \(i\) is
\begin{equation}\label{eq:destination}
    M_i^{(k)}=
    \sum_{n=1}^N\lambda_{n,\eps}^{(k)}B_n^{(k)}(z_i^{(k)}),
\end{equation}
and the full update is
\begin{equation}\label{eq:full-update}
    z_i^{(k+1)}=M_i^{(k)},
    \qquad i=1,\ldots,m.
\end{equation}
The main algorithm therefore has no damping or line-search parameter. Proposition~\ref{prop:minimizer} shows that it is the exact minimizer of the MM surrogate, not a heuristic relaxation.

\begin{algorithm}[ht]
\caption{Direct fixed-weight free-support Wasserstein--Weiszfeld solver}\label{alg:direct}
\begin{algorithmic}[1]
\REQUIRE Empirical measures \(\mu_n=\sum_j w_{n,j}\delta_{x_{n,j}}\); outer weights \(\pi\); fixed support weights \(v\); support size \(m\); smoothing \(\eps\ge0\); tolerance.
\ENSURE Support locations \(Z\) defining \(\barmu_Z=\sum_i v_i\delta_{z_i}\).
\STATE Initialize \(Z^{(0)}\), for example by \(k\)-means on the pooled support points.
\FOR{\(k=0,1,2,\ldots\)}
    \FOR{\(n=1,\ldots,N\)}
        \STATE Solve exact OT between \(\barmu_{Z^{(k)}}\) and \(\mu_n\), obtaining \(\Gamma^{(n,k)}\).
        \STATE Compute \(d_n^{(k)}=W_2(\barmu_{Z^{(k)}},\mu_n)\) and \(r_{n,\eps}^{(k)}=\sqrt{(d_n^{(k)})^2+\eps^2}\).
        \STATE Compute barycentric projections \(B_n^{(k)}(z_i^{(k)})\) by \eqref{eq:baryproj}.
    \ENDFOR
    \STATE Compute Weiszfeld weights \(\lambda_{n,\eps}^{(k)}\) by \eqref{eq:lambda}.
    \STATE Set \(z_i^{(k+1)}=\sum_n\lambda_{n,\eps}^{(k)}B_n^{(k)}(z_i^{(k)})\) for all \(i\).
    \STATE Stop if the objective or residual change is below tolerance.
\ENDFOR
\end{algorithmic}
\end{algorithm}

When \(\eps=0\), Algorithm~\ref{alg:direct} is the nonsmoothed Wasserstein--Weiszfeld iteration whenever all distances \(d_n^{(k)}\) are positive. If a current candidate coincides with an input measure, the inverse-distance weights become singular, as in the classical Euclidean Weiszfeld problem. For this reason, the implemented default is \(\eps>0\). In the experiments, nonsmoothed runs use a deterministic floor rule whenever a distance falls below a fixed numerical threshold.

\section{Majorization--minimization analysis}\label{sec:theory}

This section gives the solver guarantees. The results are intentionally stated for \(\eps>0\). This is not because the nonsmoothed iteration is uninteresting, but because smoothing removes singular inverse-distance weights and makes all surrogate slopes finite.

For a selected optimal plan \(\Gamma^{(n,k)}\), define the frozen-plan quadratic cost
\begin{equation}\label{eq:Qnk}
    Q_{n,k}(Z)=
    \sum_{i=1}^m\sum_{j=1}^{m_n}
    \Gamma_{ij}^{(n,k)}\norm{z_i-x_{n,j}}^2.
\end{equation}
At the current iterate set
\begin{equation}\label{eq:cnk}
    c_n^{(k)}=Q_{n,k}(Z^{(k)})=W_2^2(\barmu_{Z^{(k)}},\mu_n).
\end{equation}
We use \(c_n^{(k)}\) for this scalar to avoid overloading the support-size notation.

\subsection{The surrogate and its minimizer}

The first proposition is the reason the method is an MM algorithm. The transport plan optimal at \(Z^{(k)}\) is still feasible at any other \(Z\), and the square root in the smoothed median objective has a supporting tangent.

\begin{proposition}[Smoothed frozen-plan majorization]\label{prop:majorization}
Fix \(\eps>0\). At iteration \(k\), let \(\Gamma^{(n,k)}\) be any selected exact optimal plans. Define
\begin{equation}\label{eq:surrogate}
    S_{k,\eps}(Z)=
    \sum_{n=1}^N \pi_n
    \left[
    r_{n,\eps}^{(k)}+
    \frac{Q_{n,k}(Z)-c_n^{(k)}}{2r_{n,\eps}^{(k)}}
    \right].
\end{equation}
Then
\[
    \calM_{m,\eps}(Z)\le S_{k,\eps}(Z)
    \quad\text{for all }Z,
    \qquad
    \calM_{m,\eps}(Z^{(k)})=S_{k,\eps}(Z^{(k)}).
\]
No uniqueness of the exact OT plans is required.
\end{proposition}

The next proposition explains why the update has no step size. The quadratic surrogate in Proposition~\ref{prop:majorization} has a closed-form minimizer.

\begin{proposition}[Closed-form surrogate minimizer]\label{prop:minimizer}
Let
\begin{equation}\label{eq:Ak}
    A_k=\frac12\sum_{n=1}^N\frac{\pi_n}{r_{n,\eps}^{(k)}}.
\end{equation}
Up to a constant independent of \(Z\), the surrogate satisfies
\begin{equation}\label{eq:square}
    S_{k,\eps}(Z)=C_k+A_k\sum_{i=1}^m v_i\norm{z_i-M_i^{(k)}}^2,
\end{equation}
where \(M_i^{(k)}\) is defined in \eqref{eq:destination}. Hence the unique minimizer of \(S_{k,\eps}\) is \(z_i=M_i^{(k)}\), \(i=1,\ldots,m\).
\end{proposition}

\subsection{Descent and residual rate}

The following theorem is the basic solver guarantee. It says that every exact plan selection gives valid descent. Degeneracy of the exact OT LP may change the trajectory, but it does not invalidate the descent inequality.

\begin{theorem}[Monotone descent]\label{thm:descent}
Let \(\eps>0\), and suppose the update \eqref{eq:full-update} is computed using exact OT plans. Then
\begin{equation}\label{eq:descent}
    \calM_{m,\eps}(Z^{(k+1)})
    \le
    \calM_{m,\eps}(Z^{(k)})
    -
    A_k
    \sum_{i=1}^m v_i\norm{z_i^{(k)}-M_i^{(k)}}^2.
\end{equation}
Consequently, \(\{\calM_{m,\eps}(Z^{(k)})\}\) is non-increasing and converges.
\end{theorem}

The compactness needed for residual results is automatic. Since each barycentric projection lies in the convex hull of the corresponding target support, and the Weiszfeld weights form a convex combination, the updated support stays in the convex hull of all input atoms.

\begin{lemma}[Convex-hull invariance]\label{lem:hull}
Let
\[
    K_X=\conv\{x_{n,j}:n=1,\ldots,N,\ j=1,\ldots,m_n\}.
\]
If \(Z^{(0)}\in K_X^m\), then \(Z^{(k)}\in K_X^m\) for all \(k\). Even if \(Z^{(0)}\notin K_X^m\), the full update satisfies \(Z^{(1)}\in K_X^m\), and hence \(Z^{(k)}\in K_X^m\) for all \(k\ge1\).
\end{lemma}

Let
\[
    R_k=\sum_{i=1}^m v_i\norm{z_i^{(k)}-M_i^{(k)}}^2
\]
be the fixed-point residual. Lemma~\ref{lem:hull} yields a uniform lower bound on \(A_k\) after the first iteration. This gives both asymptotic residual convergence and a finite-time stationarity rate.

\begin{corollary}[Residual convergence and best-residual rate]\label{cor:residual}
Assume \(\eps>0\). Let \(D_X=\max_{x,y\in K_X}\norm{x-y}\) and \(\bar D_\eps=\sqrt{D_X^2+\eps^2}\). For all \(k\ge1\),
\[
    A_k\ge \frac{1}{2\bar D_\eps}.
\]
Moreover,
\[
    R_k\to0,
\]
and for every \(K\ge1\),
\begin{equation}\label{eq:res-rate}
    \min_{1\le k\le K} R_k
    \le
    \frac{2\bar D_\eps\{\calM_{m,\eps}(Z^{(1)})-\inf_Z\calM_{m,\eps}(Z)\}}{K}.
\end{equation}
\end{corollary}

This \(O(1/K)\) statement is not a rate for the support locations; it is a rate for approaching fixed-point stationarity as measured by the natural MM residual. Under differentiability, this residual controls the gradient as well.

\begin{corollary}[Residual-to-gradient bound]\label{cor:gradient-residual}
Suppose that, for each \(n\), the map
\[
    Z\mapsto W_2^2(\barmu_Z,\mu_n)
\]
is differentiable at \(Z^{(k)}\). Let
\[
    \sigma_k=\sum_{n=1}^N\frac{\pi_n}{r_{n,\eps}^{(k)}}=2A_k,
    \qquad
    v_{\max}=\max_i v_i.
\]
Then
\begin{equation}\label{eq:grad-residual}
    \norm{\nabla \calM_{m,\eps}(Z^{(k)})}^2
    \le \sigma_k^2 v_{\max} R_k
    \le \frac{v_{\max}}{\eps^2}R_k .
\end{equation}
Consequently, a small residual implies approximate first-order stationarity at points where the discrete OT value functions are differentiable.
\end{corollary}

A stronger objective rate requires more local structure. The next result gives one standard non-degenerate regime in which the residual error bound used in \PLoja and \KLoja analyses \citep{polyak_1963_GradientMethodsMinimisation,lojasiewicz_1963_ProprieteTopologiqueSousensembles,attouch_2013_ConvergenceDescentMethods,bolte_2014_ProximalAlternatingLinearized} can be verified. It is not a global convergence theorem. Rather, it records what additional local geometry is sufficient for a linear rate.

\begin{corollary}[Local linear rate in a non-degenerate regime]\label{cor:local-rate}
Let \(Z^\star\) be a fixed point satisfying the stationarity condition of Corollary~\ref{cor:stationarity}. Assume that there is a neighborhood \(U\) of \(Z^\star\) such that:
\begin{enumerate}[label=(\roman{enumi})]
    \item for each \(n\), the optimal plan between \(\barmu_Z\) and \(\mu_n\) is locally unique and constant as a function of \(Z\in U\);
    \item \(\calM_{m,\eps}\) is twice continuously differentiable on \(U\);
    \item \(Z^\star\) is a strict local minimizer and \(\nabla^2\calM_{m,\eps}(Z^\star)\) is positive definite.
\end{enumerate}
If the iterates eventually remain in a sufficiently small neighborhood of \(Z^\star\), then there exist constants \(C>0\) and \(\theta\in(0,1)\) such that
\[
    \calM_{m,\eps}(Z^{(k)})-\calM_{m,\eps}(Z^\star)
    \le C\theta^{k}
\]
for all sufficiently large \(k\).
\end{corollary}

Condition (i) is a strong non-degeneracy condition: it says that the exact OT LPs stay in the same active region near the fixed point. In that case the OT values are local quadratic functions of the support coordinates, and the smoothed median objective behaves like a smooth finite-dimensional objective with a non-degenerate local minimum. The eventual-neighborhood assumption is the usual local-rate hypothesis: it restricts the statement to the basin of attraction of the fixed point. One may alternatively replace it by an explicit initialization condition once a local basin is verified for a particular problem. Without such local geometry, the unconditional guarantee remains the best-residual bound \eqref{eq:res-rate}.

\subsection{Limit points and stationarity}

The residual convergence alone need not imply convergence of the entire support sequence. The next result describes all accumulation points.

\begin{proposition}[Limit-point self-consistency]\label{prop:limit}
Let \(\eps>0\). Every accumulation point \(Z^\infty\) of the iterates admits, for each \(n\), an exact optimal plan \(\Gamma^{(n,\infty)}\in\Pi(v,w_n)\) between \(\barmu_{Z^\infty}\) and \(\mu_n\) such that
\begin{equation}\label{eq:limit-fp}
    z_i^\infty=
    \sum_{n=1}^N\lambda_{n,\eps}^\infty
    \left(\frac{1}{v_i}\sum_{j=1}^{m_n}\Gamma_{ij}^{(n,\infty)}x_{n,j}\right),
    \qquad i=1,\ldots,m,
\end{equation}
where
\[
    \lambda_{n,\eps}^\infty=
    \frac{\pi_n/r_{n,\eps}^\infty}
    {\sum_{\ell=1}^N\pi_\ell/r_{\ell,\eps}^\infty},
    \qquad
    r_{n,\eps}^\infty=\sqrt{W_2^2(\barmu_{Z^\infty},\mu_n)+\eps^2}.
\]
\end{proposition}

Self-consistency is weaker than differentiable stationarity because the discrete OT value function may not be differentiable when multiple optimal plans induce different projected destinations. The next result gives the standard differentiable case using Danskin's envelope theorem logic. If the pointwise minimum of smooth linear-in-plan costs is differentiable, then all active optimal plans induce the same gradient \citep{danskin_1967_TheoryMaxMinIts}.

\begin{corollary}[Conditional stationarity]\label{cor:stationarity}
Let \(Z^\infty\) satisfy Proposition~\ref{prop:limit}. Suppose that for every \(n\), the map
\[
    Z\mapsto W_2^2(\barmu_Z,\mu_n)
\]
is differentiable at \(Z^\infty\). Then \(Z^\infty\) is a stationary point of \(\calM_{m,\eps}\).
\end{corollary}

\subsection{Smoothing, stability, and resolution consistency}

Smoothing is a technical device with a uniform objective error.

\begin{proposition}[Smoothing error and near-minimizers]\label{prop:smoothing}
For every \(Z\) and every \(\eps\ge0\),
\begin{equation}\label{eq:smoothing-bound}
    0\le \calM_{m,\eps}(Z)-\calM_m(Z)\le \eps.
\end{equation}
Let \(\mathcal D\subset(\R^d)^m\) be compact. If \(\eps_\ell\downarrow0\), \(\delta_\ell\downarrow0\), and \(Z_\ell\in\mathcal D\) satisfies
\[
    \calM_{m,\eps_\ell}(Z_\ell)
    \le
    \inf_{Z\in\mathcal D}\calM_{m,\eps_\ell}(Z)+\delta_\ell,
\]
then every accumulation point of \(\{Z_\ell\}\) is a minimizer of \(\calM_m\) over \(\mathcal D\).
\end{proposition}

\begin{corollary}[Fixed-smoothing near-minimizers]\label{cor:fixed-eps}
Let \(\mathcal D\subset(\R^d)^m\) be any domain and fix \(\eps\ge0\). If \(\hat Z\in\mathcal D\) is a \(\delta\)-near-minimizer of \(\calM_{m,\eps}\) over \(\mathcal D\), that is
\[
    \calM_{m,\eps}(\hat Z)\le \inf_{Z\in\mathcal D}\calM_{m,\eps}(Z)+\delta,
\]
then \(\hat Z\) is a \((\delta+\eps)\)-near-minimizer of the exact nonsmoothed objective \(\calM_m\) over the same domain:
\[
    \calM_m(\hat Z)\le \inf_{Z\in\mathcal D}\calM_m(Z)+\delta+\eps.
\]
\end{corollary}

This fixed-\(\eps\) statement is the practical reason for using a very small smoothing level. Solving the smoothed problem accurately also gives a controlled near-minimizer of the exact objective, with a deterministic additive error no larger than \(\eps\).

The median objective is stable under perturbations of the input measures. This is an objective-value statement, not a stability theorem for minimizers.

\begin{proposition}[Objective stability]\label{prop:stability}
For any fixed candidate \(\zeta\in\calP_2(\R^d)\) and any two input collections \(\mu_1,\ldots,\mu_N\) and \(\nu_1,\ldots,\nu_N\),
\[
    \left|
    \sum_{n=1}^N\pi_n W_2(\zeta,\mu_n)
    -
    \sum_{n=1}^N\pi_n W_2(\zeta,\nu_n)
    \right|
    \le
    \sum_{n=1}^N\pi_n W_2(\mu_n,\nu_n).
\]
\end{proposition}

Finally, we record the resolution statement. It concerns exact or near-exact minimizers over increasingly rich fixed-weight or other discrete approximation classes. It does not say that a nonconvex solver run at a fixed \(m\) finds the full free-support median over arbitrary weights. Before stating the general theorem, we close the loop for the approximation class used by the algorithm.

\begin{proposition}[Density of uniform empirical classes]\label{prop:uniform-density}
Let \(K\subset\R^d\) be compact. For every \(\mu\in\calP(K)\), there exists a sequence of uniform empirical measures
\[
    \nu_m=\frac1m\sum_{i=1}^m\delta_{z_i^{(m)}},
    \qquad z_i^{(m)}\in K,
\]
such that \(W_2(\nu_m,\mu)\to0\). Consequently, the fixed-uniform-weight class used by the algorithm satisfies the density hypothesis below.
\end{proposition}

\begin{theorem}[Resolution consistency]\label{thm:resolution}
Let \(K\subset\R^d\) be compact and let \(\Med\) denote the minimizer set of \(\calM\) over \(\calP(K)\). Let \(\calA_m(K)\subset\calP(K)\) be approximation classes such that for every \(\mu^\star\in\Med\),
\[
    \inf_{\nu\in\calA_m(K)} W_2(\nu,\mu^\star)\to0.
\]
Define
\[
    \beta_m=\inf_{\nu\in\calA_m(K)} \sum_{n=1}^N\pi_n W_2(\nu,\mu_n).
\]
Then \(\beta_m\to \inf_{\mu\in\calP(K)}\calM(\mu)\). If \(\hat\mu_m\in\calA_m(K)\) are near-minimizers satisfying
\[
    \calM(\hat\mu_m)-\beta_m\to0,
\]
then
\[
    \dist_{W_2}(\hat\mu_m,\Med)\to0.
\]
If the median is unique, then \(W_2(\hat\mu_m,\mu^\star)\to0\).
\end{theorem}

By Proposition~\ref{prop:uniform-density}, the theorem applies in particular to the uniform-weight particle classes used by the proposed method. This theorem is complementary to the statistical consistency results in \citet{you_2025_WassersteinMedianProbability}. Here the inputs are fixed, the approximation class becomes richer, and statistical consistency studies random empirical inputs and their limiting population median.

\section{Implementation}\label{sec:implementation}

The direct and nested solvers use the same exact OT subproblem. At iteration \(k\), each input measure requires solving
\[
    \min_{\Gamma\in\Pi(v,w_n)}\sum_{i,j}\Gamma_{ij}\norm{z_i^{(k)}-x_{n,j}}^2.
\]
The cost matrix requires \(O(d m m_n)\) operations to form, and storing the plan requires \(O(m m_n)\) memory. Barycentric projections are linear in the number of plan entries. The exact OT solver is therefore the dominant cost.

\begin{table}[ht]
\centering
\caption{Exact OT subproblem accounting for one outer median iteration. Here \(L_k\) is the number of inner free-support barycenter iterations used by a nested method at outer step \(k\).}
\label{tab:cost}
\begin{tabular}{lll}
\toprule
Method & Main operation & \makecell{Exact OT subproblems \\per outer iteration} \\
\midrule
Direct & One projection relocation & \(N\) \\
Nested-\(L\) & \(L\) inner barycenter iterations & \(N L\) \\
Nested-tight & Inner barycenter to tolerance & \(N L_k\) \\
\bottomrule
\end{tabular}
\end{table}

The Wasserstein medoid is also used in some experiments as an input-restricted robust reference. It is not an iterative free-support method, so it is not included in Table~\ref{tab:cost} since computing it requires a one-time set of pairwise exact OT distances among the input measures.

The primary diagnostics include the nonsmoothed median objective \(\calM_m\), the smoothed objective \(\calM_{m,\eps}\), the residual \(R_k\), the effective weights \(\lambda_{n,\eps}^{(k)}\), runtime, and total exact OT subproblem count. The effective weights are especially useful in contamination experiments because outlying measures should receive smaller weights if the median is behaving robustly.

For reproducibility, experiments should fix the exact OT solver, deterministic initialization, and a deterministic solver tie-rule when available. The descent theorem is invariant to which exact optimal plan is selected, but the support trajectory may depend on the selected vertex in degenerate transport LPs. Reporting final-objective sensitivity to solver tie rules is therefore useful in degenerate examples.

\section{Experiments}\label{sec:experiments}

The experiments evaluate the proposed method as a solver for empirical fixed-weight free-support Wasserstein median problems. The Wasserstein median itself is not new; the computational question is whether the direct Wasserstein--Weiszfeld relocation reaches nearly the same objective quality as a well-solved exact nested Weiszfeld method while using fewer exact OT subproblems. All comparisons in this section use exact discrete OT subproblems. No entropic or Sinkhorn barycenter routine is used.

Unless stated otherwise, all free-support methods use the same initialization, the same candidate support size, uniform candidate weights, the same exact OT solver, and the same stopping tolerance. The direct method is compared with nested exact Weiszfeld methods in which each outer Weiszfeld step calls an exact OT free-support barycenter solver for a fixed number of inner iterations, or until a tight inner tolerance is reached. The Wasserstein medoid is included as an input-restricted robust reference, and the Wasserstein barycenter is included only in contamination and application experiments as a mean-like comparison. The main reported quantities are wall-clock time, final nonsmoothed median objective, relative objective gap, total number of exact OT subproblems, outer iterations, inner iterations for nested methods, and final fixed-point residual. Objective gaps are computed relative to the best nonsmoothed median objective observed among the compared methods in the same replicate.

The experiments were run with deterministic seeds indexed by configuration and repetition. For each experiment, the replication code saves raw solver objects, objective histories, residual histories, final Weiszfeld weights, summary tables, and plotting data. The absolute runtimes reported below depend on the workstation and implementation; the more portable computational measure is the total number of exact OT subproblems.

\subsection{Exact solver efficiency}\label{sec:exp-efficiency}

The first set of experiments compares the direct method with exact nested Weiszfeld solvers on synthetic empirical measures. Each input measure is a uniformly weighted empirical measure generated from a Gaussian mixture. The base setting uses \(N=10\) input measures, input support size \(m_n=100\), median support size \(m=50\), and dimension \(d=2\). We compare Direct, Nested-2, Nested-5, Nested-10, Nested-tight, and Medoid. The nested methods use the same outer Weiszfeld weights as the direct method, but update the candidate support by approximately solving the weighted free-support barycenter subproblem at each outer iteration.

\begin{table}[t]
\centering
\small
\setlength{\tabcolsep}{4pt}
\renewcommand{\arraystretch}{1.12}
\caption{Base synthetic benchmark. Values are averages over repetitions. The gap is the relative nonsmoothed median objective gap relative to the best method in each replicate. OT denotes the total number of exact OT subproblems.}
\label{tab:exp-base-solver}
\begin{tabular}{lrrrrrr}
\toprule
Method &
\makecell{Runtime\\(sec)} &
Objective &
Gap &
\makecell{Outer\\iter.} &
\makecell{Inner\\iter.} &
\makecell{OT\\solves} \\
\midrule
Direct       & 0.0409 & 0.5165 & \(7.99\times10^{-4}\) & 8.8 & --   & 98  \\
Nested-2     & 0.0616 & 0.5163 & \(6.74\times10^{-4}\) & 6.8 & 13.6 & 146 \\
Nested-5     & 0.1480 & 0.5162 & \(5.55\times10^{-4}\) & 6.6 & 33.0 & 340 \\
Nested-10    & 0.3087 & 0.5162 & \(5.44\times10^{-4}\) & 6.8 & 68.0 & 690 \\
Nested-tight & 0.1106 & 0.5162 & \(5.44\times10^{-4}\) & 6.8 & 24.0 & 250 \\
Medoid       & 0.0366 & 0.5317 & \(1.61\times10^{-2}\) & --  & --   & 45  \\
\bottomrule
\end{tabular}
\end{table}

Table~\ref{tab:exp-base-solver} shows the basic trade-off. The direct method is slightly less accurate than the tightly solved nested method in this small benchmark, but its final objective is very close: the mean relative gap is \(7.99\times10^{-4}\), compared with \(5.44\times10^{-4}\) for Nested-tight. The direct method uses 98 exact OT subproblems on average, while Nested-tight uses 250, and Nested-10 uses 690. In wall-clock time, Direct is about 2.7 times faster than Nested-tight and about 7.6 times faster than Nested-10 in this base setting.

Figure~\ref{fig:exp-objective-time} shows objective and residual histories in the base setting. The direct method makes most of its objective progress quickly and reaches a small objective gap at substantially lower wall-clock time than the high-inner-budget nested methods. The residual panel shows that the direct method also reaches the numerical residual tolerance, although the nested methods reduce the residual more rapidly per outer iteration because each outer step contains additional inner barycenter work.

\begin{figure}[t]
\centering
\includegraphics[width=.95\textwidth]{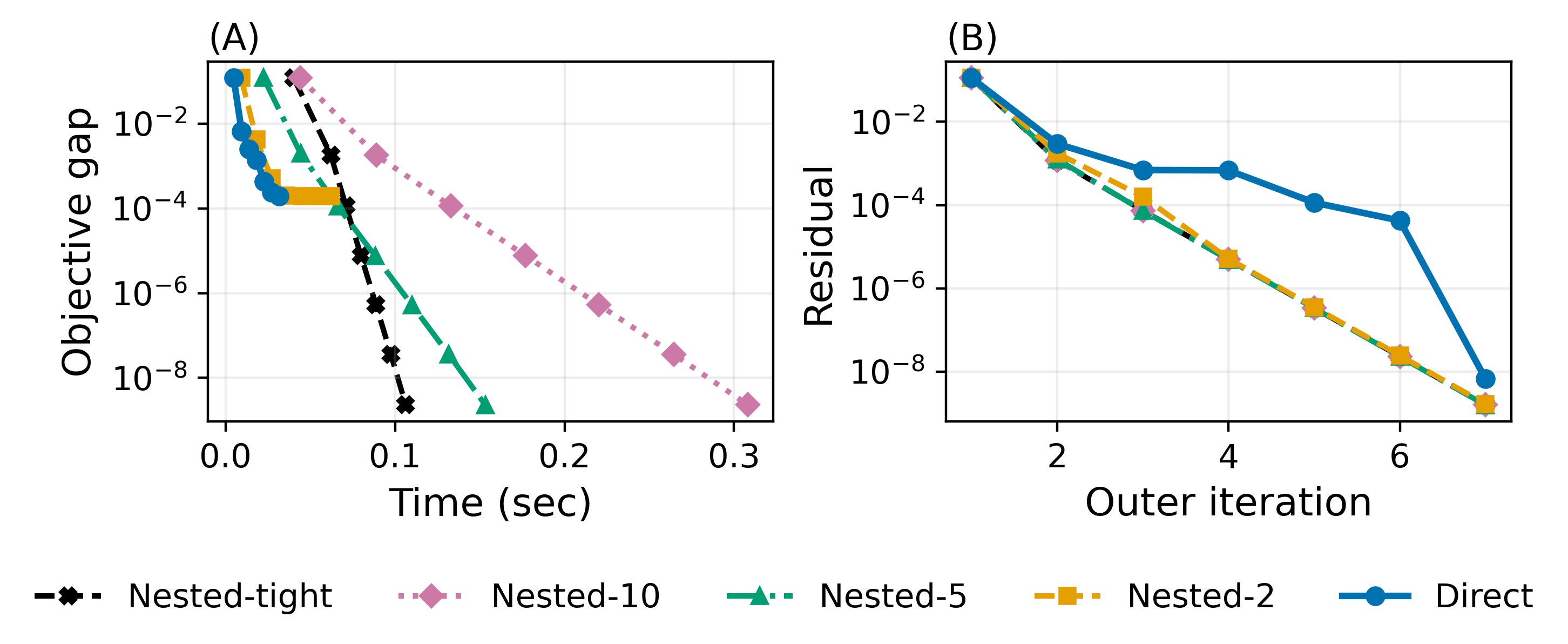}
\caption{Base synthetic benchmark. (A) Relative nonsmoothed median objective gap versus wall-clock time. (B) Fixed-point residual versus outer iteration.}
\label{fig:exp-objective-time}
\end{figure}

To assess scaling, we vary one problem parameter at a time: the number of input measures \(N\), the input support size \(m_n\), the candidate support size \(m\), and the dimension \(d\). Figure~\ref{fig:exp-scaling-runtime} summarizes runtime. Across all sweeps, Direct has the lowest runtime among the free-support median solvers. The gap between Direct and the nested methods is most visible as \(N\), \(m_n\), or \(m\) increases, which is consistent with the cost accounting in Section~\ref{sec:implementation} that nested methods repeatedly solve weighted barycenter subproblems, whereas the direct method performs one barycentric relocation per outer iteration. The corresponding exact OT subproblem counts show the same qualitative pattern and are reported in the supplementary material.

\begin{figure}[t]
\centering
\includegraphics[width=.88\textwidth]{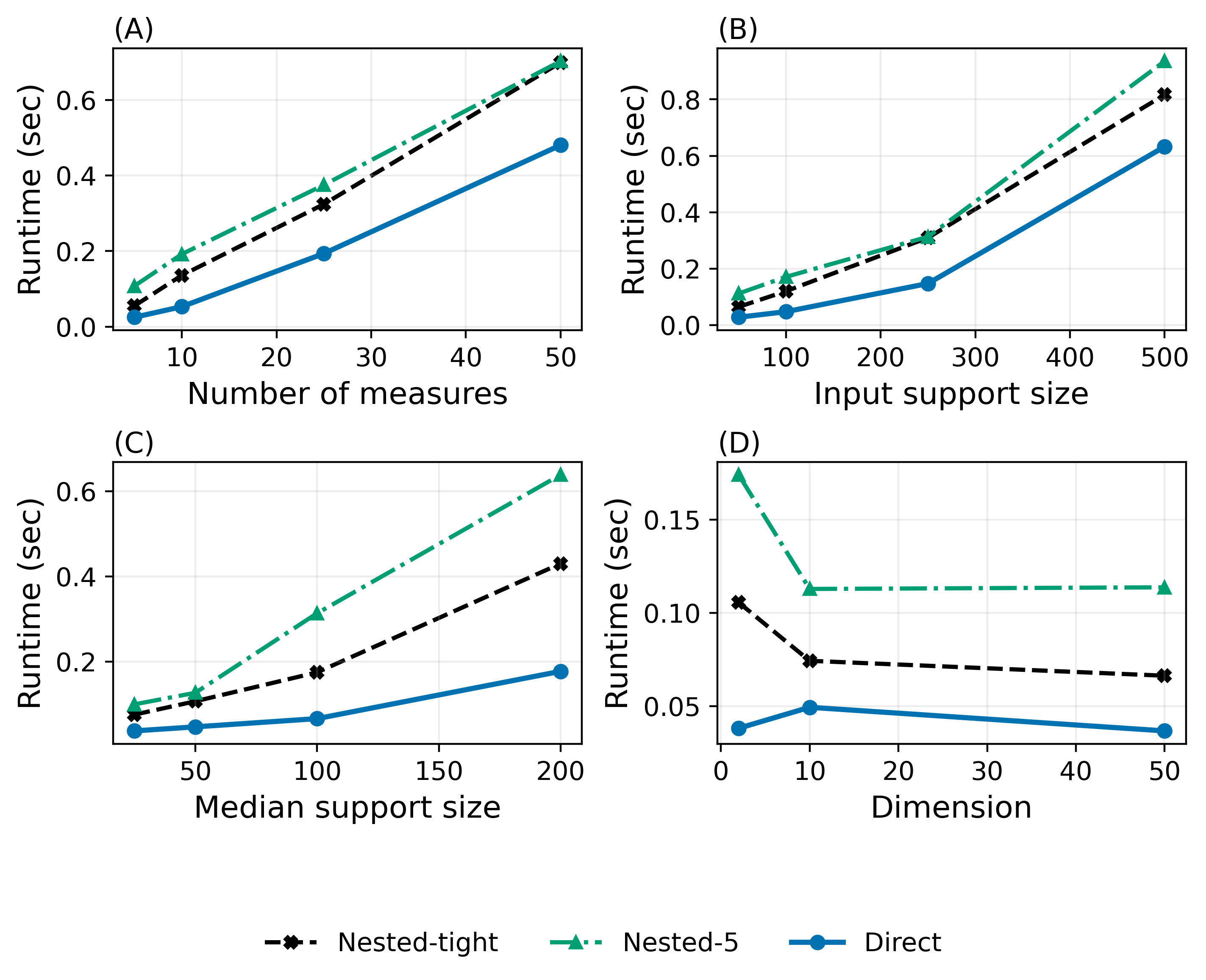}
\caption{Synthetic scaling benchmark: runtime as a function of (A) number of measures, (B) input support size, (C) median support size, and (D) dimension.}
\label{fig:exp-scaling-runtime}
\end{figure}

The second solver benchmark isolates the cost of the inner barycenter solve in the nested method. We use a base synthetic setting and a larger setting, then compare fixed inner budgets and tolerance-based inner stopping rules. The direct method is included as a reference because it performs the closed-form MM relocation without any inner barycenter loop.

Figure~\ref{fig:exp-inner-burden} plots relative objective gap against exact OT effort and runtime. In the base setting, increasing the fixed inner budget from 2 to 10 improves the gap from \(2.71\times10^{-4}\) to \(3.85\times10^{-6}\), but increases the OT count from 162 to 690 and the runtime from 0.070 seconds to 0.311 seconds. Increasing the budget beyond 10 gives essentially no objective improvement in this setting but continues to increase computational cost. The larger setting shows the same qualitative pattern: tighter inner solves reduce the objective gap but require many more exact OT subproblems. These results quantify the computational bottleneck that the direct method is designed to avoid.

\begin{figure}[t]
\centering
\includegraphics[width=.49\textwidth]{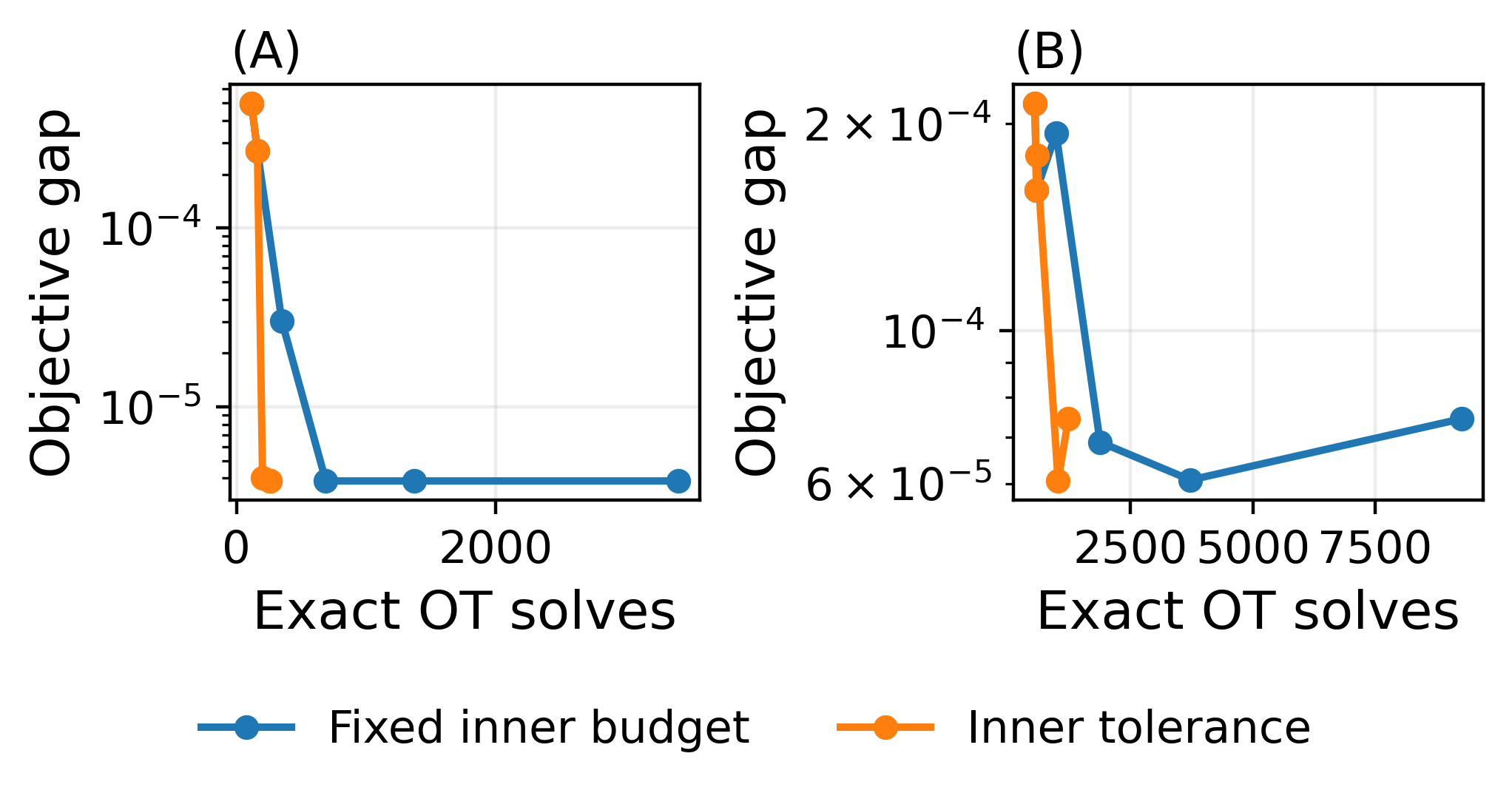}\hfill
\includegraphics[width=.49\textwidth]{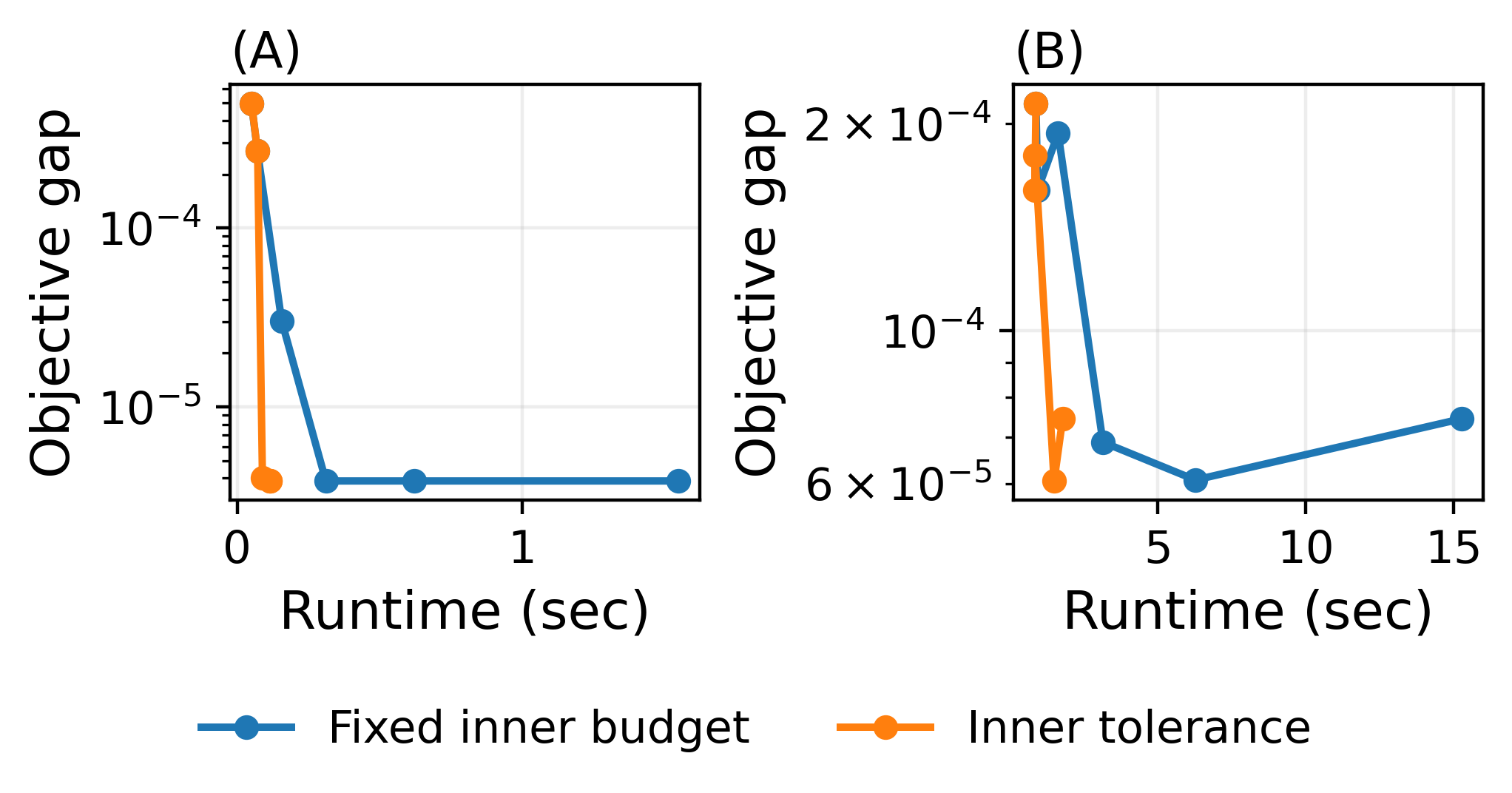}
\caption{Inner-solver burden in exact nested Weiszfeld. Left: relative objective gap versus exact OT subproblems. Right: relative objective gap versus runtime. In each figure, panel (A) is the base synthetic setting and panel (B) is the larger setting.}
\label{fig:exp-inner-burden}
\end{figure}

We also tested \(\varepsilon=\rho D_0\) with \(\rho\in\{0,10^{-10},10^{-8},10^{-6},10^{-4},10^{-2}\}\). No singularity events occurred in these runs. Small positive values were stable and had little effect on the nonsmoothed objective, while \(\rho=10^{-2}\) could alter the finite-support solution. 

\subsection{Robustness under contaminated distributions}\label{sec:exp-contamination}

The next experiment illustrates the robust behavior of the median objective while retaining the exact OT solver comparison. We generate \(N=20\) empirical Gaussian measures in \(\mathbb R^2\), contaminate a fraction \(\omega_{\mathrm{out}}\) of them by shifting their means by \(\Delta_{\mathrm{out}}\), and compare Direct, Nested-tight, Medoid, and the Wasserstein barycenter. The barycenter is included only as a mean-like reference.

Figure~\ref{fig:exp-contamination-error} shows the distance to the clean reference as the outlier fraction increases. The barycenter is strongly affected by outliers: averaged across outlier magnitudes, its distance to the clean reference increases from 0.106 with no outliers to 4.715 at 40\% outliers. The direct median increases from 0.094 to 0.610 over the same range, closely matching Nested-tight. The medoid is robust but input-restricted, and is generally less accurate than the free-support median in this experiment.

The final Weiszfeld weights provide an interpretable diagnostic of this robustness. At 40\% outliers, the total weight assigned to outlying measures is about 0.046 on average, far below their nominal 0.4 outer weight. Thus the median solver automatically downweights distant contaminated measures. Direct is also substantially cheaper than Nested-tight in this experiment: averaged over all contamination configurations, Direct takes 0.292 seconds and 318 exact OT subproblems, while Nested-tight takes 0.962 seconds and 995 exact OT subproblems, with nearly identical median objectives. The weight and runtime diagnostic plots are included in the supplementary material.

\begin{figure}[t]
\centering
\includegraphics[width=.95\textwidth]{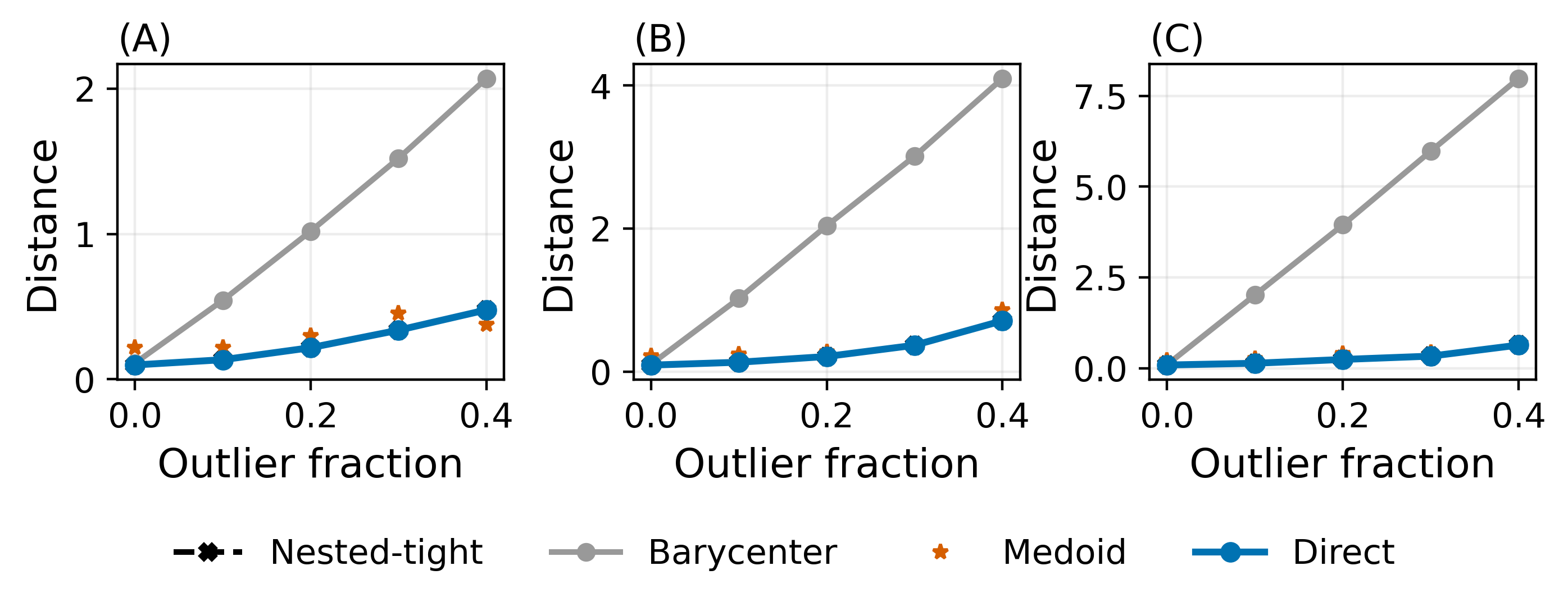}
\caption{Contaminated Gaussian experiment. Distance to the clean reference versus outlier fraction for three outlier shifts.}
\label{fig:exp-contamination-error}
\end{figure}

\subsection{Posterior aggregation and image prototypes}\label{sec:exp-applications}

We close with two statistical-computing examples. These are not intended to introduce new statistical functionals; rather, they illustrate what the exact median solver computes in two settings where robust distributional summaries are useful.

The first application is posterior aggregation. We generate Gaussian subset posterior draws from a simulated conjugate regression problem and corrupt a controlled fraction of subset posteriors by shifting or inflating them. The goal is not to benchmark MCMC, but to compare exact OT aggregation summaries after posterior draws have been obtained. We compare Direct, Nested-tight, Medoid, and Wasserstein barycenter. Accuracy is measured against the analytic full-data posterior.

With no corruption, the barycenter and median have similar posterior mean error. As corruption increases, the barycenter deteriorates more quickly. At 30\% corrupted subsets, the barycenter posterior mean error is 0.053 on average, while Direct and Nested-tight are both 0.032. Covariance and predictive RMSE show the same qualitative pattern. Direct and Nested-tight have nearly identical statistical accuracy, but Direct is faster: averaged over the posterior aggregation grid, Direct takes 0.309 seconds and 180 exact OT subproblems, while Nested-tight takes 0.511 seconds and 269 exact OT subproblems. The total Weiszfeld weight assigned to corrupted subset posteriors is also below their nominal fraction: at 30\% corruption, Direct assigns total weight about 0.149 to corrupted subsets.

Figure~\ref{fig:exp-posterior-supports} visualizes one representative 30\% corruption setting with \(K=10\) subsets and median support size \(m=100\). The direct and nested median supports are visually close to each other and remain concentrated around the analytic posterior contour, while the medoid is constrained to an input posterior draw cloud. The full posterior metric and weight curves are reported in the supplementary material.

\begin{figure}[t]
\centering
\includegraphics[width=.95\textwidth]{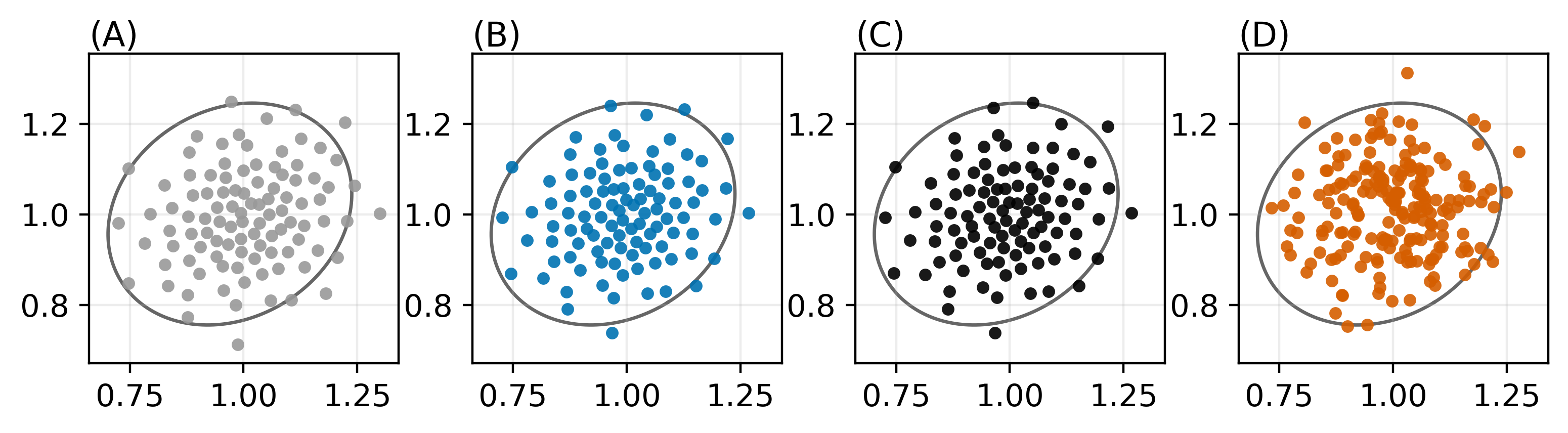}
\caption{Representative posterior supports under 30\% corrupted subset posteriors with \(K=10\) and \(m=100\). Panels show (A) barycenter, (B) Direct median, (C) Nested-tight median, and (D) medoid. The gray ellipse is the analytic posterior 95\% contour.}
\label{fig:exp-posterior-supports}
\end{figure}

The second application computes robust image prototypes from point-cloud representations of handwritten digits. We consider three target/outlier pairs: digit 0 contaminated by digit 6, digit 1 contaminated by digit 7, and digit 8 contaminated by digit 3. Images are converted to point clouds by retaining foreground pixel coordinates. We compare Direct, Nested-tight, Medoid, and Wasserstein barycenter as the outlier fraction increases.

Direct and Nested-tight are nearly indistinguishable across all three digit pairs. Both are slightly more stable than the barycenter, and both are substantially closer to the clean reference than the medoid. Averaged over the full image grid, Direct has distance 0.093 to the clean reference, compared with 0.093 for Nested-tight, 0.098 for barycenter, and 0.151 for medoid. Direct is also faster: 0.177 seconds on average versus 0.277 seconds for Nested-tight. At 40\% contamination, the average total outlier weight is about 0.366 for Direct, below the nominal contamination level.

Figure~\ref{fig:exp-image-prototypes} displays representative prototypes at the largest contamination level. Rows correspond to the target/outlier pairs 0/6, 1/7, and 8/3; columns show the clean reference, barycenter, Direct median, Nested-tight median, and medoid. The direct and nested median prototypes are visually similar, while the medoid remains constrained to an input-restricted representative. The clean-reference distance curves and outlier-weight curves are included in the supplementary material.

\begin{figure}[t]
\centering
\includegraphics[width=.92\textwidth]{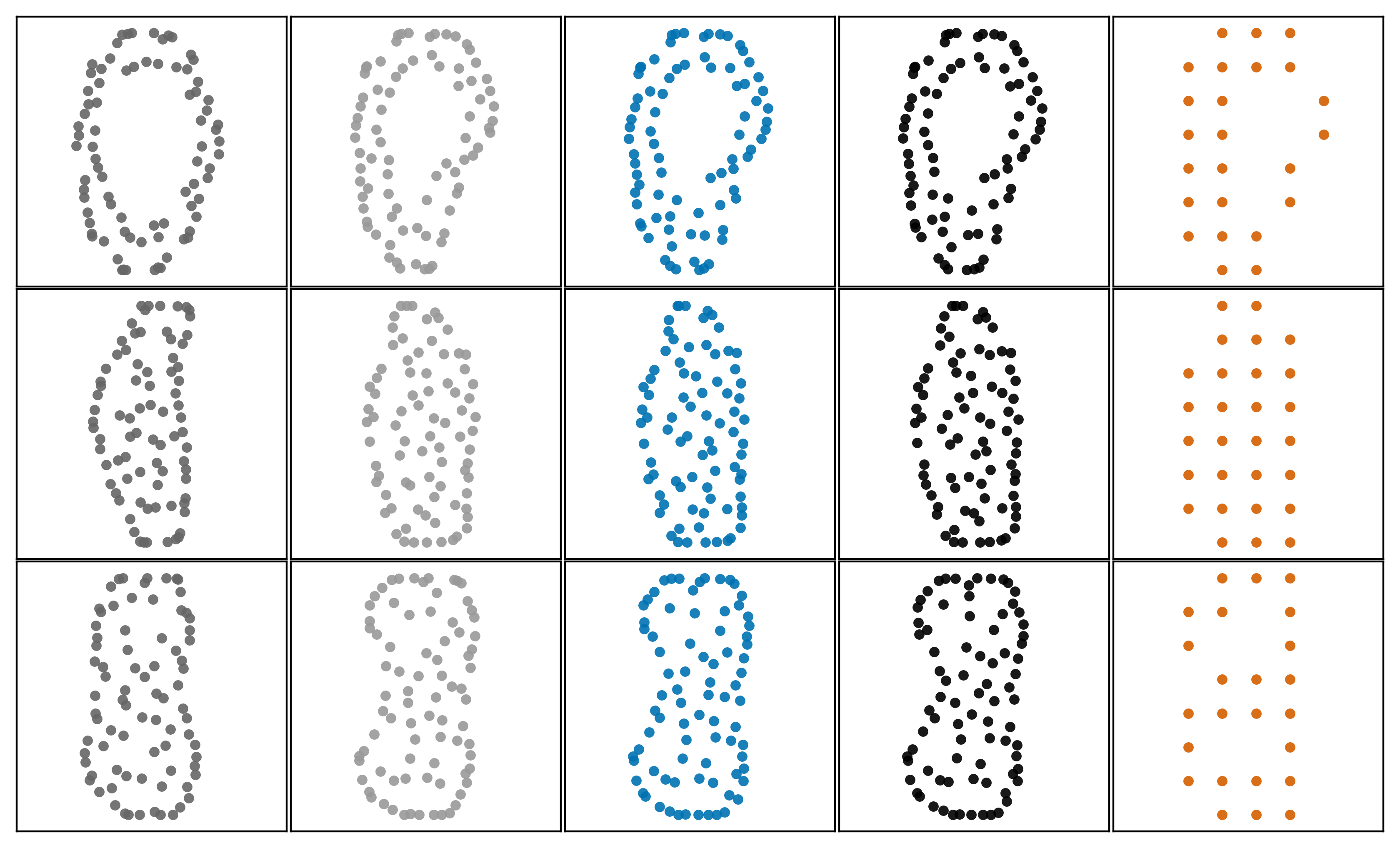}
\caption{Image prototypes at 40\% contamination and support size \(m=80\). Rows correspond to target/outlier digit pairs 0/6, 1/7, and 8/3. Columns show clean reference, barycenter, Direct median, Nested-tight median, and medoid.}
\label{fig:exp-image-prototypes}
\end{figure}

\subsection{Summary of empirical findings}\label{sec:exp-summary}

The experiments support the computational motivation of the direct solver. In the synthetic base benchmark, the direct method reaches a median objective very close to Nested-tight while using about 39\% of its exact OT subproblems and about 37\% of its runtime. The inner-burden experiment shows that increasing the accuracy of the nested inner barycenter solve improves the final objective but rapidly increases exact OT effort. The smoothing experiment shows that small positive smoothing is numerically stable and that the nonsmoothed iteration is well defined in the tested random instances. In the contaminated Gaussian, posterior aggregation, and image prototype experiments, Direct and Nested-tight produce nearly identical median summaries, while Direct consistently uses fewer exact OT subproblems. The barycenter comparisons confirm the expected mean-like sensitivity to outlying distributions, and the final Weiszfeld weights provide an interpretable diagnostic of how the median downweights contaminated measures.

\section{Discussion}\label{sec:discussion}

This paper develops a direct fixed-weight free-support solver for empirical Wasserstein medians. The Wasserstein median itself is a known robust Fr\'echet functional; the contribution here is computational. The standard metric-space Weiszfeld approach updates the current candidate by solving a weighted Wasserstein barycenter problem, so each outer median iteration contains an inner free-support barycenter solve. The proposed method removes that inner loop. It uses the exact OT plans already needed to evaluate the current distances, forms barycentric projections, and performs one closed-form Wasserstein--Weiszfeld relocation.

The main theoretical point is that this relocation is not merely an early-stopped barycenter iteration. For the smoothed objective \(\calM_{m,\eps}\), freezing the selected exact OT plans and linearizing the square-root distance produce a tight quadratic MM surrogate, whose minimizer is exactly the proposed update. This gives monotone descent for any selected exact optimal plans, even when the transport LP is degenerate and optimal plans are non-unique. The theory also explains the natural residual diagnostic, shows that the iterates remain in the convex hull of the input supports, and gives a finite-time best-residual rate. The residual-to-gradient bound and conditional stationarity result clarify when the fixed-point residual can be interpreted as ordinary first-order stationarity.

The experiments support the intended computational role of the method. In the synthetic benchmark, the direct solver reaches a median objective close to a tightly solved nested method while using fewer exact OT subproblems. The inner-burden experiment shows why this matters: improving the inner barycenter solve in the nested method can reduce the median objective, but it rapidly increases exact OT effort. In the contaminated Gaussian, posterior aggregation, and image-prototype experiments, the direct solver produces median summaries nearly indistinguishable from the nested exact solver, while the barycenter comparisons show the expected sensitivity of mean-like summaries to outlying distributions. The final Weiszfeld weights are also useful diagnostics, because they reveal which input measures are downweighted by the median computation.

The method has several limitations. First, the solver optimizes only support locations with fixed atom weights. This fixed-weight formulation is useful computationally and is dense in the large-\(m\) limit, but it is not the full finite-\(m\) free-support problem over both weights and locations. A natural extension is a block-coordinate method that alternates between support relocation and weight updates. Second, the support-location problem is nonconvex, so initialization can affect the final local solution. The descent theory guarantees objective decrease and fixed-point self-consistency, not global optimality at a fixed support size. Third, the method uses repeated exact OT solves. This keeps the comparison with nested exact Weiszfeld clean, but it also limits scalability when the number of measures or support sizes become large. Approximate OT variants would be useful in larger problems, provided their errors are controlled within an inexact-MM framework.

Several directions follow from these limitations. On the algorithmic side, adaptive support weights, deterministic tie-breaking for degenerate transport plans, and stochastic or mini-batch variants would broaden the range of usable problem sizes. On the theoretical side, sharper local convergence results could be obtained under explicit non-degeneracy or \KLoja conditions, while statistical-computing applications would benefit from finite-sample error decompositions separating sampling error, support-resolution error, smoothing error, and optimization error. Finally, robust posterior aggregation and distributional clustering provide natural settings in which Wasserstein medians may be preferable to barycenters, but the practical choice between them should be guided by the observed effective weights, objective histories, and downstream performance.








\begin{appendices}

\section{Proofs of theoretical results}\label{app:proof}
\subsection{Proof of Proposition~\ref{prop:majorization} }
\begin{proof}
For each \(n\), \(\Gamma^{(n,k)}\) is feasible for every support configuration because \(\Pi(v,w_n)\) depends only on the weights. Hence
\[
    W_2^2(\barmu_Z,\mu_n)\le Q_{n,k}(Z),
\]
with equality at \(Z^{(k)}\). The function \(s\mapsto\sqrt{s+\eps^2}\) is increasing and concave on \([0,\infty)\). Therefore
\[
\sqrt{W_2^2(\barmu_Z,\mu_n)+\eps^2}
\le
\sqrt{Q_{n,k}(Z)+\eps^2}
\le
r_{n,\eps}^{(k)}+\frac{Q_{n,k}(Z)-c_n^{(k)}}{2r_{n,\eps}^{(k)}}.
\]
Multiplying by \(\pi_n\) and summing gives the result. Equality at \(Z^{(k)}\) follows because both inequalities are tight there.
\end{proof}

\subsection{Proof of Proposition~\ref{prop:minimizer}}
\begin{proof}
Expanding \(Q_{n,k}\),
\[
Q_{n,k}(Z)=
\sum_i v_i\norm{z_i}^2
-2\sum_i\left\langle z_i,\sum_j\Gamma_{ij}^{(n,k)}x_{n,j}\right\rangle
+C_{n,k},
\]
where \(C_{n,k}\) is independent of \(Z\). The \(Z\)-dependent part of \(S_{k,\eps}\) is
\[
A_k\sum_i v_i\norm{z_i}^2
-
2\sum_i\left\langle z_i,
\sum_n\frac{\pi_n}{2r_{n,\eps}^{(k)}}\sum_j\Gamma_{ij}^{(n,k)}x_{n,j}
\right\rangle.
\]
Since \(\lambda_{n,\eps}^{(k)}=(\pi_n/(2r_{n,\eps}^{(k)}))/A_k\), the linear coefficient equals \(A_k v_iM_i^{(k)}\). Completing the square gives \eqref{eq:square}. Since \(A_k>0\) and \(v_i>0\), the minimizer is unique.
\end{proof}

\subsection{Proof of Theorem~\ref{thm:descent}}
\begin{proof}
By Proposition~\ref{prop:majorization},
\[
\calM_{m,\eps}(Z^{(k+1)})\le S_{k,\eps}(Z^{(k+1)}).
\]
By Proposition~\ref{prop:minimizer},
\[
S_{k,\eps}(Z^{(k+1)})-S_{k,\eps}(Z^{(k)})
=
-A_k\sum_i v_i\norm{z_i^{(k)}-M_i^{(k)}}^2.
\]
Since \(S_{k,\eps}(Z^{(k)})=\calM_{m,\eps}(Z^{(k)})\), the descent inequality follows. Nonnegativity of the objective gives convergence of the objective values.
\end{proof}

\subsection{Proof of Lemma~\ref{lem:hull}}
\begin{proof}
For each \(n\) and \(i\), the coefficients \(\Gamma_{ij}^{(n,k)}/v_i\) are nonnegative and sum to one, so \(B_n^{(k)}(z_i^{(k)})\in\conv\{x_{n,j}:j=1,\ldots,m_n\}\subset K_X\). The weights \(\lambda_{n,\eps}^{(k)}\) are nonnegative and sum to one. Hence \(M_i^{(k)}\in K_X\). Since the full update sets \(z_i^{(k+1)}=M_i^{(k)}\), the claim follows.
\end{proof}

\subsection{Proof of Corollary~\ref{cor:residual}}
\begin{proof}
By Lemma~\ref{lem:hull}, for all \(k\ge1\), \(Z^{(k)}\in K_X^m\). Thus every distance \(W_2(\barmu_{Z^{(k)}},\mu_n)\) is at most \(D_X\), and \(r_{n,\eps}^{(k)}\le\bar D_\eps\). Therefore
\[
A_k=\frac12\sum_n\frac{\pi_n}{r_{n,\eps}^{(k)}}\ge\frac{1}{2\bar D_\eps}.
\]
Summing Theorem~\ref{thm:descent} from \(k=1\) to \(K\) gives
\[
\sum_{k=1}^{K} A_kR_k
\le
\calM_{m,\eps}(Z^{(1)})-\calM_{m,\eps}(Z^{(K+1)}).
\]
Letting \(K\to\infty\) shows \(\sum_k R_k<\infty\), hence \(R_k\to0\). The finite-time bound follows from the same inequality and the lower bound on \(A_k\).
\end{proof}

\subsection{Proof of Corollary~\ref{cor:gradient-residual}}
\begin{proof}
For each differentiable OT value, Danskin's theorem gives the same gradient for every selected optimal plan:
\[
\nabla_{z_i}W_2^2(\barmu_{Z^{(k)}},\mu_n)
=
2v_i\{z_i^{(k)}-B_n^{(k)}(z_i^{(k)})\}.
\]
By the chain rule,
\[
\nabla_{z_i}\calM_{m,\eps}(Z^{(k)})
=
 v_i \sigma_k\left(z_i^{(k)}-M_i^{(k)}\right),
\]
because \(M_i^{(k)}\) is the \(\lambda_{n,\eps}^{(k)}\)-weighted average of the barycentric projections and \(\lambda_{n,\eps}^{(k)}=(\pi_n/r_{n,\eps}^{(k)})/\sigma_k\). Therefore
\[
\norm{\nabla \calM_{m,\eps}(Z^{(k)})}^2
= \sigma_k^2\sum_{i=1}^m v_i^2\norm{z_i^{(k)}-M_i^{(k)}}^2
\le \sigma_k^2 v_{\max} R_k.
\]
Since \(r_{n,\eps}^{(k)}\ge \eps\) and \(\sum_n\pi_n=1\), we also have \(\sigma_k\le1/\eps\). This proves the stated bound.
\end{proof}

\subsection{Proof of Corollary~\ref{cor:local-rate}}
\begin{proof}
Under the local active-plan condition, each squared OT value is represented in \(U\) by a fixed quadratic polynomial in \(Z\). Hence \(\calM_{m,\eps}\) is twice continuously differentiable in a possibly smaller neighborhood of \(Z^\star\). By positive definiteness of the Hessian at \(Z^\star\), after shrinking the neighborhood if necessary there is a constant \(\ell_H>0\) such that \(\nabla^2\calM_{m,\eps}(Z)\succeq \ell_H I\) throughout the neighborhood. Thus the standard local Polyak--\L{}ojasiewicz inequality holds there:
\[
    \calM_{m,\eps}(Z)-\calM_{m,\eps}(Z^\star)
    \le \frac{1}{2\ell_H}\norm{\nabla\calM_{m,\eps}(Z)}^2.
\]
Let \(\Sigma_U=\sup_{Z\in U}\sum_n\pi_n/r_{n,\eps}(Z)\), which is finite because \(\eps>0\), and let \(v_{\max}=\max_i v_i\). Corollary~\ref{cor:gradient-residual} gives
\[
    \calM_{m,\eps}(Z)-\calM_{m,\eps}(Z^\star)
    \le \kappa R(Z),
    \qquad
    \kappa=\frac{\Sigma_U^2 v_{\max}}{2\ell_H},
\]
where \(R(Z)=\sum_i v_i\|z_i-M_i(Z)\|^2\). On the same compact neighborhood, \(A(Z)=\frac12\sum_n\pi_n/r_{n,\eps}(Z)\) has a positive lower bound \(\underline A\). Enlarging \(\kappa\) if necessary so that \(\kappa>\underline A\), the descent inequality gives, for all sufficiently large \(k\),
\[
\calM_{m,\eps}(Z^{(k+1)})-\calM_{m,\eps}(Z^\star)
\le
\left(1-\frac{\underline A}{\kappa}\right)
\{\calM_{m,\eps}(Z^{(k)})-\calM_{m,\eps}(Z^\star)\}.
\]
Iterating the contraction proves the claim.
\end{proof}

\subsection{Proof of Proposition~\ref{prop:limit}}
\begin{proof}
Let \(Z^{(k_\ell)}\to Z^\infty\). The transport polytopes are compact, so after passing to a subsequence we may assume \(\Gamma^{(n,k_\ell)}\to\Gamma^{(n,\infty)}\) for each \(n\). Optimality passes to the limit because the transport cost is continuous in \(Z\). Therefore \(\Gamma^{(n,\infty)}\) is optimal at \(Z^\infty\). Since \(\eps>0\), the weights \(\lambda_{n,\eps}^{(k_\ell)}\) converge to \(\lambda_{n,\eps}^\infty\). Finally, \(R_{k_\ell}\to0\) implies \(z_i^{(k_\ell)}-M_i^{(k_\ell)}\to0\). Passing to the limit in \eqref{eq:destination} gives \eqref{eq:limit-fp}.
\end{proof}

\subsection{Proof of Corollary~\ref{cor:stationarity}}
\begin{proof}
For fixed \(n\), the discrete OT value is the minimum over a fixed compact polytope of a function smooth in \(Z\). By Danskin's theorem, differentiability at \(Z^\infty\) implies that every optimal plan induces the same gradient. Hence using \(\Gamma^{(n,\infty)}\),
\[
\nabla_{z_i}W_2^2(\barmu_Z,\mu_n)\big|_{Z=Z^\infty}
=
2v_i\left(z_i^\infty-\frac{1}{v_i}\sum_j\Gamma_{ij}^{(n,\infty)}x_{n,j}\right).
\]
The chain rule gives
\[
\nabla_{z_i}\calM_{m,\eps}(Z^\infty)
=
 v_i\sum_n\frac{\pi_n}{r_{n,\eps}^\infty}
\left(z_i^\infty-\frac{1}{v_i}\sum_j\Gamma_{ij}^{(n,\infty)}x_{n,j}\right).
\]
Factoring out \(\sum_n\pi_n/r_{n,\eps}^\infty\), the remaining term is the difference between \(z_i^\infty\) and the right side of the fixed-point equation in Proposition~\ref{prop:limit}. Thus the gradient is zero.
\end{proof}

\subsection{Proof of Proposition~\ref{prop:smoothing}}
\begin{proof}
For \(d\ge0\),
\[
0\le \sqrt{d^2+\eps^2}-d\le\eps.
\]
The bound \eqref{eq:smoothing-bound} follows by applying this inequality to each term and summing with weights \(\pi_n\). Let \(m_\eps^{\mathcal D}=\inf_{Z\in\mathcal D}\calM_{m,\eps}(Z)\) and \(m^{\mathcal D}=\inf_{Z\in\mathcal D}\calM_m(Z)\). The uniform bound gives \(m_\eps^{\mathcal D}\le m^{\mathcal D}+\eps\), while \(\calM_m\le\calM_{m,\eps}\). If \(Z_\ell\) is \(\delta_\ell\)-near-minimizing for \(\calM_{m,\eps_\ell}\) over \(\mathcal D\), then
\[
\calM_m(Z_\ell)\le m^{\mathcal D}+\eps_\ell+\delta_\ell.
\]
Continuity of \(\calM_m\) on compact \(\mathcal D\) implies that every accumulation point minimizes \(\calM_m\) over \(\mathcal D\).
\end{proof}

\subsection{Proof of Corollary~\ref{cor:fixed-eps}}
\begin{proof}
Let \(Z_m\) be any exact or approximate minimizer of \(\calM_m\) over \(\mathcal D\). By Proposition~\ref{prop:smoothing},
\[
\inf_{Z\in\mathcal D}\calM_{m,\eps}(Z)
\le
\inf_{Z\in\mathcal D}\calM_m(Z)+\eps.
\]
If \(\hat Z\) is \(\delta\)-near-minimizing for \(\calM_{m,\eps}\), then
\[
\calM_m(\hat Z)
\le
\calM_{m,\eps}(\hat Z)
\le
\inf_{Z\in\mathcal D}\calM_{m,\eps}(Z)+\delta
\le
\inf_{Z\in\mathcal D}\calM_m(Z)+\eps+\delta,
\]
which is the claim.
\end{proof}

\subsection{Proof of Proposition~\ref{prop:stability}}
\begin{proof}
The reverse triangle inequality for the metric \(W_2\) gives
\[
\left|W_2(\zeta,\mu_n)-W_2(\zeta,\nu_n)\right|\le W_2(\mu_n,\nu_n).
\]
Multiplying by \(\pi_n\) and summing yields the result.
\end{proof}

\subsection{Proof of Proposition~\ref{prop:uniform-density}}
\begin{proof}
Fix \(\mu\in\calP(K)\). Since finitely supported probability measures are dense in \((\calP(K),W_2)\), choose
\[
    \eta=\sum_{\ell=1}^{L} a_\ell\delta_{y_\ell}
\]
with \(y_\ell\in K\), \(a_\ell>0\), \(\sum_\ell a_\ell=1\), and \(W_2(\eta,\mu)\) arbitrarily small. For each large \(m\), choose nonnegative integers \(k_{\ell,m}\) summing to \(m\) such that \(k_{\ell,m}/m\to a_\ell\). Define the uniform empirical measure
\[
    \eta_m=\frac1m\sum_{\ell=1}^{L}\sum_{r=1}^{k_{\ell,m}}\delta_{y_\ell}.
\]
After merging repeated atoms, \(\eta_m\) is still represented by \(m\) equally weighted particles, and \(\eta_m\to\eta\) weakly. Since the support is compact, this convergence is also in \(W_2\). The triangle inequality gives
\[
    W_2(\eta_m,\mu)\le W_2(\eta_m,\eta)+W_2(\eta,\mu).
\]
First take \(m\to\infty\), and then choose \(\eta\) closer to \(\mu\). This proves the density of uniform empirical measures.
\end{proof}

\subsection{Proof of Theorem~\ref{thm:resolution}}
\begin{proof}
The functional \(\calM\) is continuous on \(\calP(K)\) under \(W_2\), and \(\calP(K)\) is compact. Let \(m^*=\inf_{\mu\in\calP(K)}\calM(\mu)\). Since \(\calA_m(K)\subset\calP(K)\), \(\beta_m\ge m^*\). By density, for any \(\mu^\star\in\Med\) there exists \(\nu_m\in\calA_m(K)\) with \(W_2(\nu_m,\mu^\star)\to0\). Continuity gives \(\calM(\nu_m)\to m^*\), hence \(\limsup\beta_m\le m^*\). Thus \(\beta_m\to m^*\).

If \(\hat\mu_m\) are near-minimizers, compactness gives convergent subsequences. Along any subsequence \(\hat\mu_{m_\ell}\to\eta\), continuity gives \(\calM(\eta)=m^*\), so \(\eta\in\Med\). If \(\dist_{W_2}(\hat\mu_m,\Med)\) did not converge to zero, a subsequence bounded away from \(\Med\) would have a further convergent subsequence with limit in \(\Med\), a contradiction. The unique-median statement follows immediately.
\end{proof}

\end{appendices}

\bibliography{references}

\end{document}